\documentclass[fleqn,11pt,twoside]{article}

\usepackage{amsthm,amsthm,amssymb, color, xcolor,epsfig, graphics, subfigure}

\usepackage{amsmath, graphicx, latexsym, lscape}

\makeatletter
\newcommand{\copyrightnote}[2]{{\renewcommand{\thefootnote}{}
 \footnotetext{\small\it
\begin{flushleft}
 \copyright \ #1   #2  
\end{flushleft}}}}

\newcommand{\Name}[1]{\begin{flushleft}
                       \LARGE \bf #1
                       \end{flushleft}\vspace{-3mm}}

\newcommand{\Author}[1]{\begin{flushleft}
                       \it #1 \end{flushleft}}

\newcommand{\Address}[1]{\begin{flushleft}
                       \it #1 \end{flushleft}}

\newcommand{\Date}[1]{\begin{flushleft}
                      \small  \it #1 \end{flushleft}}

\newcommand{\evenhead}{Author \ name}
\newcommand{\oddhead}{Article \ name}

\renewcommand{\@evenhead}{
\hspace*{-3pt}\raisebox{-15pt}[\headheight][0pt]{\vbox{\hbox to \textwidth
{\thepage \hfil \evenhead}\vskip4pt \hrule}}}
\renewcommand{\@oddhead}{
\hspace*{-3pt}\raisebox{-15pt}[\headheight][0pt]{\vbox{\hbox to \textwidth
{\oddhead \hfil \thepage}\vskip4pt\hrule}}}
\renewcommand{\@evenfoot}{}
\renewcommand{\@oddfoot}{}

\long\def\@makecaption#1#2{%
  \vskip\abovecaptionskip
  \sbox\@tempboxa{\small \textbf{#1.}\ \ #2}%
  \ifdim \wd\@tempboxa >\hsize
    {\small \textbf{#1.}\ \ #2}\par
  \else
    \global \@minipagefalse
    \hb@xt@\hsize{\hfil\box\@tempboxa\hfil}%
  \fi
  \vskip\belowcaptionskip}

\newcommand{\JNMPnumberwithin}[3][\arabic]{%
  \@ifundefined{c@#2}{\@nocounterr{#2}}{%
    \@ifundefined{c@#3}{\@nocnterr{#3}}{%
      \@addtoreset{#2}{#3}%
      \@xp\xdef\csname the#2\endcsname{%
        \@xp\@nx\csname the#3\endcsname .\@nx#1{#2}}}}%
}

\renewenvironment{proof}[1][\proofname]{\par
  \normalfont
  \topsep6\p@\@plus6\p@ \trivlist
  \item[\hskip\labelsep\textbf{%
    #1\@addpunct{.}}]\ignorespaces
}{%
  \qed\endtrivlist
}

\newcommand{\resetfootnoterule} {
  \renewcommand\footnoterule{%
  \kern-3\p@
  \hrule\@width.4\columnwidth
  \kern2.6\p@}
}

\renewcommand{\footnoterule}{}

\makeatother

\theoremstyle{definition}
\newtheorem*{definition}{Definition}
\newtheorem{remark}{Remark}

\theoremstyle{plain}
\newtheorem{proposition}{Proposition}

\begin{document}

\renewcommand{\evenhead}{ {\LARGE\textcolor{blue!10!black!40!green}{{\sf \ \ \ ]ocnmp[}}}\strut\hfill V E Adler}
\renewcommand{\oddhead}{ {\LARGE\textcolor{blue!10!black!40!green}{{\sf ]ocnmp[}}}\ \ \ \ \  Negative flows and non-autonomous reductions of the Volterra lattice}

\thispagestyle{empty}
\newcommand{\FistPageHead}[3]{
\begin{flushleft}
\raisebox{8mm}[0pt][0pt]
{\footnotesize \sf
\parbox{150mm}{{Open Communications in Nonlinear Mathematical Physics}\ \ \ \ {\LARGE\textcolor{blue!10!black!40!green}{]ocnmp[}}
\quad Special Issue 1, 2024\ \  pp
#2\hfill {\sc #3}}}\vspace{-13mm}
\end{flushleft}}

\FistPageHead{1}{\pageref{firstpage}--\pageref{lastpage}}{ \ \ }

\strut\hfill

\strut\hfill

\copyrightnote{The author(s). Distributed under a Creative Commons Attribution 4.0 International License}

\begin{center}
{  {\bf This article is part of an OCNMP Special Issue\\ 
\smallskip
in Memory of Professor Decio Levi}}
\end{center}

\smallskip

\Name{Negative flows and non-autonomous reductions of the Volterra lattice}

\Author{V.E.\:Adler $^{1,2}$}

\Address{$^1$ L.D.\:Landau Institute for Theoretical Physics, Akademika Semenova av. 1A, 142432, Chernogolovka, Russian Federation\\ $^2$ Institute of Mathematics, Ufa Federal Research Centre, Russian Academy of Sciences, Chernyshevsky str. 112, 450008, Ufa, Russian Federation}

\Date{Received July 18, 2023; Accepted December 11, 2023}

\setcounter{equation}{0}

\begin{abstract}

\noindent 
We study reductions of the Volterra lattice corresponding to stationary equations for the additional, noncommutative subalgebra of symmetries. It is shown that, in the case of general position, such a reduction is equivalent to the stationary equation for a sum of the scaling symmetry and the negative flows, and is written as $(m+1)$-component difference equations of the Painlev\'e type generalizing the dP$_1$ and dP$_{34}$ equations. For these reductions, we present the isomonodromic Lax pairs and derive the B\"acklund transformations which form the $\mathbb{Z}^m$ lattice.

\end{abstract}

\label{firstpage}


\section{Introduction}

The Volterra lattice
\begin{equation}\label{i.ux}
 u_{n,x}=u_n(u_{n+1}-u_{n-1})
\end{equation}
is, along with the Toda lattice, one of the oldest integrable models with discrete variables \cite{Manakov_1974, Kac_Moerbeke_1975}, the study of which significantly expanded the scope of the inverse scattering method. Great merit in the development of the theory of differential-difference equations belongs to Decio Levi. In his pioneering work \cite{Levi_1981} it was shown that such equations can be naturally interpreted as the B\"acklund transformations for partial differential equations. The lattice equation (\ref{i.ux}) admits the symmetry
\begin{equation}\label{i.ut}
 u_{n,t}=u_n(u_{n+1}(u_{n+2}+u_{n+1}+u_n)-u_{n-1}(u_n+u_{n-1}+u_{n-2}))
\end{equation}
(which means that this pair of equations is consistent). Due to equation (\ref{i.ux}), all variables $u_{n\pm i}$ can be expressed in terms of a pair of neighboring variables $u=u_n$, $v=u_{n+1}$ and their $x$-derivatives, hence equation (\ref{i.ut}) turns out to be equivalent to some PDE system for $u$ and $v$, namely, 
\begin{equation}\label{Levi}
 u_t=-u_{xx}+(2uv+u^2)_x,\quad v_t=v_{xx}+(2uv+v^2)_x.
\end{equation}
Since $n$ is arbitrary, the mapping $(u_n,u_{n+1})\mapsto(u_{n+1},u_{n+2})$ defined by equation (\ref{i.ux}), serves as the simplest B\"acklund transformation for this system. A number of works by Levi and his coauthors were devoted to classification of integrable lattice equations, description of the algebra of their higher and classical symmetries, methods for constructing exact solutions, and connections with other classes of integrable equations; we only mention the papers \cite{Levi_Yamilov_1999, Levi_Winternitz_Yamilov_2001} and the book by D. Levi, P. Winternitz and R. Yamilov \cite{Levi_Winternitz_Yamilov_2022} which summed up their many years of work in the field of discrete equations.

In this paper, we consider the problem of constructing reductions of the Volterra lattice, that is, the constraints which are consistent with the shift operator $T:n \mapsto n+1$. In particular, we demonstrate that equations (\ref{i.ux}) and (\ref{i.ut}) admit a family of finite-dimensional non-autonomous reductions defined by the following $(m+1)$-component system of difference equations: 
\begin{equation}\label{nar-xt}
\left\{\begin{aligned}
 & u_n(y^j_{n+1}+y^j_n)(y^j_n+y^j_{n-1})=\alpha^j(y^j_n)^2+(-1)^n\beta^jy^j_n+\gamma^j,\qquad j=1,\dots,m, \\[3pt]
 & 4tu_n(u_{n+1}+u_n+u_{n-1})+2xu_n+n-\delta+(-1)^n\varepsilon=y^1_n+\dots+y^m_n
\end{aligned}\right.
\end{equation}
where $\alpha^j$, $\beta^j$, $\gamma^j$, $\delta$ and $\varepsilon$ are arbitrary constants. Here the first $m$ equations play the role of definition of the auxiliary non-local variables $y^1_n,\dots,y^m_n$, and the last equation defines the constraint for $u_n$. The main proposition is that this system is compatible with derivations with respect to the parameters $x$ and $t$ defined by (\ref{i.ux}), (\ref{i.ut}) and the consistent rules for differentiating $y^j_n$ (given by equations (\ref{yx}), (\ref{yt}) in the next section).

For the simplect particular case $m=0$ studied in \cite{Fokas_Its_Kitaev_1991}, the system (\ref{nar-xt}) consists of a single equation for $u_n$
\begin{equation}\label{dP1}
 4tu_n(u_{n+1}+u_n+u_{n-1})+2xu_n+n-\delta+(-1)^n\varepsilon=0, 
\end{equation}
which coincides up to dilations with the discrete Painlev\'e equation dP$_1$ \cite[eq.\,(51)]{Grammaticos_Ramani_2014}, while the evolution with respect to the parameter $x$ is governed by the continuous P$_4$ equation. For the case $m=1$ and $t=0$ studied in \cite{Adler_Shabat_2019}, $u_n$ is expressed in terms of $y_n$ from the second equation (\ref{nar-xt}) and the first one takes the form
\begin{equation}\label{dP34}
 (y^1_{n+1}+y^1_n)(y^1_n+y^1_{n-1})
  = 2x\frac{\alpha(y^1_n)^2+(-1)^n\beta y^1_n+\gamma_n}{y^1_n-n+\delta-(-1)^n\varepsilon}.
\end{equation}
This is equivalent to the discrete Painlev\'e equation dP$_{34}$ \cite[eqs.\,(13), (45)]{Grammaticos_Ramani_2014}, while the evolution with respect to $x$ is governed by P$_5$ for $\alpha\ne0$ and P$_3$ for $\alpha=0$. Non-Abelian versions of (\ref{dP1}) and (\ref{dP34}) were proposed in \cite{Adler_2020a}.

Both equations (\ref{dP1}) and (\ref{dP34}), as well as the general system (\ref{nar-xt}) belong to the so-called string equations, that is, stationary equations for non-autonomous symmetries of the Volterra lattice. In particular, it is easy to see that (\ref{dP1}) is obtained by integrating from the stationary equation for the classical scaling symmetry which is defined as
\[
 u_{n,\tau_1} = 2tu_{n,t}+xu_{n,x}+u_n.
\]
Similarly, in \cite{Adler_Shabat_2019} it was shown that equation (\ref{dP34}) can be obtained from the stationary equation for the master-symmetry of the lattice equation (\ref{i.ux}) which is of the form \cite{Fuchssteiner_1983, Oevel_Zhang_Fuchssteiner_1989, Cherdantsev_Yamilov_1995, Levi_Winternitz_Yamilov_2022}
\[
 u_{n,\tau_2}= xu_{n,t}+u_n((n+3)u_{n+1}+u_n-nu_{n-1}),
\]
although here the passage to (\ref{dP34}) is less obvious and requires some change of variables. Our goal in this article is to generalize this change for the string equations of arbitrary order.

Let us outline the content of the work. The main results are obtained in Section \ref{s:nar}, where it is shown that the system (\ref{nar-xt}) (and its further extensions) is equivalent in the case of general position to the stationary equation for higher non-autonomous symmetries. On the other hand, the system (\ref{nar-xt}) defines the stationary equation for the sum of the scaling symmetry and the negative symmetries of the Volterra lattice, which are defined in terms of the variables $y^j_n$. We mention papers \cite{Orlov_Rauch-Wojciechowski_1993, Hone_Kuznetsov_Ragnisco_1999, Adler_Kolesnikov_2023} where similar results were obtained for the Korteweg--de Vries (KdV) equation.

The negative flows are introduced in Section \ref{s:neg} based on the notion of the recursion operator \cite{Oevel_Zhang_Fuchssteiner_1989, Khanizadeh_Mikhailov_Wang_2013}. Many of the formulas in this section are fairly standard and easy to find in the literature, in particular, in the books \cite{Suris_2003, Levi_Winternitz_Yamilov_2022}. Negative symmetry can be interpreted as a generating function for higher symmetries, and the same is true for the corresponding matrices from the Lax representations. This idea itself is not new and goes back to \cite{Gelfand_Dikii}, where the flows of the KdV hierarchy are derived from the resolvent of the Sturm--Liouville operator by expansion with respect to the spectral parameter. At the same time, the author is not aware of works where the definition of the negative symmetry for the Volterra hierarchy would be given in the general form. Note that our Definition \ref{def:neg} includes, as a special case, the negative flow studied in \cite{Pritula_Vekslerchik_2003, Hu_Xue_2003, Zhu_Tam_2004}.

The final Section \ref{s:DBT} contains the Darboux--B\"acklund transformations, in the form of involutive rational mappings that preserve the form of the system (\ref{nar-xt}) and are consistent with the $x$- and $t$-evolution.

\section{Negative flows of the Volterra lattice}\label{s:neg}

Let us recall that the hierarchy of the Volterra lattice
\begin{equation}\label{ux}
 u_{n,x}=u_n(u_{n+1}-u_{n-1})
\end{equation}
is generated by the recursion operator (see e.g. \cite{Oevel_Zhang_Fuchssteiner_1989, Khanizadeh_Mikhailov_Wang_2013})
\begin{equation}\label{R}
 R= u_n+u_n(u_{n+1}T^2-u_{n-1}T^{-1})(T-1)^{-1}\frac{1}{u_n},
\end{equation}
where $T:n\mapsto n+1$ is the shift operator. Applying $R$ to the right hand side of (\ref{ux}) gives the first higher symmetry
\begin{equation}\label{ut}
 u_{n,t}=u_n(u_{n+1}(u_{n+2}+u_{n+1}+u_n)-u_{n-1}(u_n+u_{n-1}+u_{n-2})),
\end{equation}
up to the addition of the term $cu_{n,x}$, where $c$ is an arbitrary constant of integration arising from the non-empty kernel of the operator $T-1$. Further application of $R$ generates an infinite sequence of derivations $D_i$ defined by equations
\begin{equation}\label{Di}
 u_{n,t_i}=R^{i-1}(u_{n,x})=u_n(T-T^{-1})(h^{(i)}_n),\quad i=1,2,\dotso.
\end{equation}
For uniformity, we use the notation $x=t_1$ and $t=t_2$ here and below. Equations (\ref{Di}) are local (that is, the right-hand sides are expressed in terms of $u_n$ variables) and are consistent with each other. More precisely, one can prove that $h^{(i)}_n$ are polynomials of degree $i$ in $u_{n-i+1},\dots,u_{n+i-1}$, homogeneous if one choose zero integration constants at each step, and satisfying the identities
\begin{equation}\label{hh}
 h^{(j)}_{n,t_i}=h^{(i)}_{n,t_j}
\end{equation}
which imply the property $[D_i,D_j]=0$.

Negative symmetries are defined by inverting the operator $R-\alpha$ with an arbitrary constant $\alpha$ (which is natural, since such an operator is also a recursion operator), that is, as a flow of the form
\begin{equation}\label{Rneg}
 u_{n,\xi}= (R-\alpha)^{-1}(u_{n,x}).
\end{equation}
Unlike higher symmetries, this equation is nonlocal: to write it explicitly, it is necessary to expand the set of dynamic variables. Let us show how this can be done. Let $g_n$ be the right hand side of (\ref{Rneg}), that is
\begin{equation}\label{Rag}
 (R-\alpha)(g_n)=u_n(u_{n+1}-u_{n-1}).
\end{equation}
To be able to apply $R$, we set $g_n=u_n(z_{n+1}-z_n)$; this gives
\[
 u^2_n(z_{n+1}-z_n)+u_n(u_{n+1}z_{n+2}-u_{n-1}z_{n-1})-\alpha u_n(z_{n+1}-z_n)=u_n(u_{n+1}-u_{n-1})
\]
and after canceling the factor $u_n$ we come to
\begin{equation}\label{uz}
 (T+1)(u_nz_{n+1}-u_{n-1}z_{n-1}-u_n+u_{n-1})-\alpha(z_{n+1}-z_n)=0.
\end{equation}
For the second term to lie in the image of $T+1$, we set $z_n-1=y_n+y_{n-1}$, then the equation is integrated once:
\begin{equation}\label{uy'}
 u_n(y_{n+1}+y_n)-u_{n-1}(y_{n-1}+y_{n-2})-\alpha(y_n-y_{n-1})+(-1)^n\beta=0
\end{equation}
(note that exactly the same equation is obtained if we start from (\ref{Rag}) with zero right hand side and apply the substitution $z_n=y_n+y_{n-1}$ instead, which amounts to changing the integration constant in $R(g_n)$). Equation (\ref{uy'}) can be integrated once again, noting that after multiplication by $y_n+y_{n-1}$ the left hand side belongs to the image of $T-1$. As a result, we arrive at the following definition.

\begin{definition}\label{def:neg}
The negative symmetry of the Volterra hierarchy is a flow of the form
\begin{equation}\label{neg}
 u_{n,\xi}= u_n(y_{n+1}-y_{n-1})
\end{equation}
where the non-local variable $y_n$ satisfies the following equations:
\begin{equation}\label{uy}
 u_n(y_{n+1}+y_n)(y_n+y_{n-1})=\alpha y^2_n+(-1)^n\beta y_n+\gamma
\end{equation}
with constants $\alpha$, $\beta$ and $\gamma$ which do not vanish simultaneously, and 
\begin{equation}\label{Diy}
 y_{n,t_i}=h^{(i)}_{n,\xi}
\end{equation}
where $h^{(i)}_n$ are homogeneous polynomials from (\ref{Di}).
\end{definition}

The rule (\ref{Diy}) defines the prolongation of derivations $D_i$ to the variables $y_n$ in a way consistent with the commutativity property of the flows (\ref{Di}) and (\ref{neg}). For instance, derivations $D_1=D_x$ and $D_2=D_t$ correspond to $h^{(1)}_n=u_n$ and $h^{(2)}_n=u_n(u_{n+1}+u_n+u_{n-1})$, hence the extensions of the flows (\ref{ux}) and (\ref{ut}) are defined by equations
\begin{gather}
\label{yx}
 y_{n,x}= u_n(y_{n+1}-y_{n-1}),\\
\label{yt}
 y_{n,t}= u_n\bigl(u_{n+1}(y_{n+2}-y_n)+(u_{n+1}+2u_n+u_{n-1})(y_{n+1}-y_{n-1})+u_{n-1}(y_n-y_{n-2})\bigr).
\end{gather}
The relation (\ref{hh}) implies that $D_i$ and $D_j$ remain commutative on the extended variable set. For the correctness of the Definition \ref{def:neg}, it is also necessary that $D_i$ be consistent with the difference equation (\ref{uy}). The check of consistency for $D_1=D_x$ and $D_2=D_t$ can be done by direct calculation.

\begin{proposition}\label{pr:neg}
Equation (\ref{uy}) is consistent with derivations $D_x$ and $D_t$ defined by (\ref{ux}), (\ref{yx}) and (\ref{ut}), (\ref{yt}).
\end{proposition}
\begin{proof}
Denote $F_n=u_n(y_{n+1}+y_n)(y_n+y_{n-1})-\alpha y^2_n-(-1)^n\beta y_n-\gamma$ and let $G_n=\dfrac{F_n-F_{n-1}}{y_n+y_{n-1}}$ be the left hand side of (\ref{uy'}). It is verified by direct calculation that
\begin{align}
\label{Fx}
 & F_{n,x} = u_n(y_n+y_{n-1})G_{n+1}+u_n(y_{n+1}+y_n)G_n,\\
\label{Ft} 
 & \begin{aligned}[b]
  F_{n,t} & = u_nu_{n+1}(y_n+y_{n-1})G_{n+2}\\
   &\qquad +u_n\bigl(u_{n-1}(y_{n-1}+y_{n-2})+2(u_{n+1}+u_n)(y_n+y_{n-1})\bigr)G_{n+1}\\
   &\qquad +u_n\bigl(u_{n+1}(y_{n+2}+y_{n+1})+2(u_n+u_{n-1})(y_{n+1}+y_n)\bigr)G_n\\
   &\qquad + u_nu_{n-1}(y_{n+1}+y_n)G_{n-1}.
 \end{aligned}  
\end{align}
The right hand sides vanish if $F_n=0$ for all $n$, as required.
\end{proof}

Note that when checking the commutativity property $[D_i,D_j]=0$, the equation (\ref{uy}) is not used. Moreover, commutativity is also preserved for more general extensions of the form $y_{n,t_i}=h^{(i)}_{n,\xi}+\delta^{(i)}_n$ where $\delta^{(i)}_{n+2}=\delta^{(i)}_n$ are ``blinking'' constants. However, for such flows, compatibility with (\ref{uy}) is satisfied only for $\delta^{(i)}_n=0$. For the examples of $D_x$ and $D_t$, this can be verified directly, similar to the above proof.

One application of the negative symmetry is that the expansion $(R-\alpha)^{-1}=\alpha^{-1}(1+R/\alpha+R^2/\alpha^2+\dots)$ 
allows interpreting the flow (\ref{neg}) as a generating function for the higher symmetries (\ref{Di}):
\begin{equation}\label{tseries}
 u_{n,\xi}= \frac{1}{\alpha}u_{n,t_1}+\frac{1}{\alpha^2}u_{n,t_2}+\frac{1}{\alpha^3}u_{n,t_3}+\dotso.
\end{equation}
In particular, this leads to the following assertion, which implies the locality property of flows (\ref{Di}) formulated above (note that their commutativity can also be proved from the consistency property of negative symmetries with different parameters $\alpha$, but we will not delve into this).

\begin{proposition}\label{pr:negR}
Homogeneous polynomials $h^{(i)}_n(u_{n-i+1},\dots,u_{n+i-1})$ defining the higher symmetries (\ref{Di}) are calculated as the coefficients of the formal series
\begin{equation}\label{hseries}
 y_n=\frac{1}{2}+\frac{h^{(1)}_n}{\alpha}+\frac{h^{(2)}_n}{\alpha^2}+\dots
\end{equation}
satisfying the equation
\begin{equation}\label{uyhom}
 u_n(y_{n+1}+y_n)(y_n+y_{n-1})=\alpha y^2_n -\frac{\alpha}{4},
\end{equation}
which is equivalent to the explicit recurrent relations
\begin{equation}\label{hrec}
 h^{(i+1)}_n = u_n\sum^i_{s=0}\bigl(h^{(s)}_{n+1}+h^{(s)}_n\bigr)
  \bigl(h^{(i-s)}_n+h^{(i-s)}_{n-1}\bigr)
  -\sum^i_{s=1}h^{(s)}_nh^{(i+1-s)}_n,\quad 
 h^{(0)}_n=\frac{1}{2}.
\end{equation}
\end{proposition}

The parameters $\beta$ and $\gamma$, which enter the defining equation (\ref{uy}) as integration constants, are less important than $\alpha$. Notice, that equations (\ref{neg})--(\ref{Diy}) keep their form under the change $y_n\mapsto y_n+(-1)^nc$ with constant $c$. If $\alpha\ne0$, this makes possible to set $\beta=0$, while the parameter $\gamma$ can be rescaled, as is done is equations (\ref{uyhom}). If $\alpha=0$ and $\beta\ne0$, then it is possible to set $\gamma=0$. Thus, all negative symmetries belong to one of three types:
\[
 \alpha\ne0,~~\beta=0;\qquad \alpha=\gamma=0,~~\beta\ne0;\qquad \alpha=\beta=0,~~\gamma\ne0,
\]
moreover, the last two types contain actually one equation each, up to the scaling.  

Since negative symmetries form a family, more complicated flows can be constructed by adding derivations $u_{n,\xi^j}= u_n(y^j_{n+1}-y^j_{n-1})$ where, for each $j$, the variables $y^j_n$ satisfy equations (\ref{uy}) with parameters $\alpha^j$, $\beta^j$ and $\gamma^j$. We use this possibility in the next section when constructing reductions of the Volterra hierarchy. In such combined flows, one can assume that $\alpha^i\ne\alpha^j$ for $i\ne j$, since the flows corresponding to one value of $\alpha$ form a linear space. Indeed, let us consider a linear combination of such flows:
\[
 u_{n,\xi}=c_1u_{n,\xi^i}+c_2u_{n,\xi^j}=u_n(y_{n+1}-y_{n-1}),\quad y_n=c_1y^i_n+c_2y^j_n;
\] 
since $y^i_n$ and $y^j_n$ satisfy equation (\ref{uy'}) with the same $\alpha=\alpha^i=\alpha^j$ ($\beta^i$ and $\beta^j$ may be different), hence $y_n$ also satisfies equation (\ref{uy'}) with the same $\alpha$ and multiplying by the integrating factor $y_{n+1}+y_n$ brings to equation of the form (\ref{uy}). 

\begin{remark}
The case $\alpha=\beta=0$ was studied in the literature, in different notations. In this case we do not need the change $z_n=y_n+y_{n-1}+1$ applied after equation (\ref{uz}). Let us set $\gamma=1$, without loss of generality, then we have the following equations for $z_n$:
\[
 u_{n,\xi}=u_n(z_{n+1}-z_n),\quad u_n(z_{n+1}-1)(z_n-1)=1.
\]
The change $z_n-1=1/p_{n-1}$ brings this to the flow
\[
 u_{n,\xi}=p_{n-1}-p_n,\quad u_n=p_{n-1}p_n
\]
introduced in the paper \cite{Pritula_Vekslerchik_2003}, see also \cite{Hu_Xue_2003, Zhu_Tam_2004}.
\end{remark}

\begin{remark}
Equation (\ref{uy}) can be solved with respect to $u_n$, which makes possible to write (\ref{yx}) as the lattice equation for the variables $y_n$ only:
\begin{equation}\label{V-CD}
 y_{n,x} = -(\alpha y^2_n+(-1)^n\beta y_n+\gamma)\left(\frac{1}{y_{n+1}+y_n}-\frac{1}{y_n+y_{n-1}}\right).
\end{equation}
Similarly, (\ref{yt}) is rewritten as a higher symmetry of this equation. For $\beta=0$ (which can be achieved in the generic case $\alpha\ne0$, as shown above), equation (\ref{V-CD}) is a special case of the V$_2$ type of the Yamilov classification \cite{Yamilov_1983, Yamilov_2006, Levi_Winternitz_Yamilov_2022}, and if $\beta\ne0$ then the additional change $y_n\mapsto(-1)^ny_n$ brings it to a special case of the V$_3$ type. For $\beta=0$, let us set $\alpha=-4a^2$ and $\gamma=c^2$, then equation (\ref{uy}) is equivalent to a composition of two discrete substitutions of the Miura type
\begin{equation}\label{Miura}
 u_n=(f_{n+1}-a)(f_n+a),\quad f_n=\frac{ay_n-ay_{n-1}+c}{y_n+y_{n-1}},
\end{equation}
moreover, the intermediate variable $f_n$ satisfies the modified Volterra lattice (type V$_1$ of the Yamilov classification, as well as the Volterra lattice itself)
\[
 f_{n,x}=(f^2_n-a^2)(f_{n+1}-f_{n-1}).
\]
Thus, the relation (\ref{uy}) between $u_n$ and $y_n$ is quite well known in the theory of the Volterra type lattices. The only new (to the author's knowledge) assertion is that the extension of the set of dynamical variables $u_n$ with the help of this relation allows one to extend the symmetry algebra by the flow (\ref{neg}).
\end{remark}

\begin{remark}
For comparison, we write down formulae for the negative symmetry of the KdV equation
\[
 u_t=u_{xxx}-6uu_x.
\]
Comparing to the Volterra lattice, here the notation of independent variables is changed: $t$ plays the role of $x$ and $x$ plays the role of $n$. The analogs of equations (\ref{neg}), (\ref{uy}) and (\ref{yx}) are, respectively, of the form
\[
 u_\xi=y_x,\qquad 
 y_{xx}=\frac{y^2_x-c^2}{2y}+2(u-\alpha)y,\qquad
 y_t=y_{xxx}-6uy_x.
\]
These equations are consistent, that is, satisfy the identities $(u_t)_\xi=(u_\xi)_t$ and $(y_{xx})_t=(y_t)_{xx}$. It is also possible to define the rules for derivatives of $y$ in virtue of higher symmetries of KdV. The analog of the lattice (\ref{V-CD}) is obtained by elimitation of $u$ from the equation for $y_t$:
\[
 y_t=y_{xxx}-\frac{3y_xy_{xx}}{y}+\frac{3y_x(y^2_x-c^2)}{2y^2}-6\alpha y_x,
\]
which is the rational form of the degenerate Calogero--Degasperis equation \cite{Calogero_Degasperis_1981, Fokas_1980}. The equation for $y_{xx}$, solved with respect to $u$, defines the differential substitution from this equation to KdV. It is decomposed into two Miura type substitutions, with the intermediate variable satisfying the modified KdV equation:
\[
 u=f_x+f^2+\alpha,\qquad f=\frac{y_x+c}{2y},\qquad f_t=f_{xxx}-6(f^2+\alpha)f_x.
\]
These negative symmetries can be used to construct multifield Painlev\'e equations which define the KdV reductions \cite{Adler_Kolesnikov_2023}.
\end{remark}

\begin{remark}
The passage from the lattice equations (\ref{i.ux}) and (\ref{i.ut}) to the Levi system (\ref{Levi}) is easily generalized for the negative symmetry. Indeed, equation (\ref{uy}) makes possible to express all $y_{n\pm k}$ in terms of $y_n=p$, $y_{n+1}=q$ and $u_n=u$, $u_{n+1}=v$ at two arbitrary neighboring lattice sites. By expressing  $y_{n-1}$ and $y_{n+2}$ in terms of these variables and denoting $\hat\beta=(-1)^n\beta$, we arrive to the hyperbolic system
\begin{equation}\label{Levi-neg}
 u_\xi= p_x=  u(p+q)-\frac{\alpha p^2+\hat\beta p+\gamma}{p+q},\quad
 v_\xi= q_x= -v(p+q)+\frac{\alpha q^2-\hat\beta q+\gamma}{p+q},
\end{equation}
which defines the negative symmetry directly for the Levi system.
\end{remark}

To conclude this section, we present zero curvature representations for the equations under study, restricting ourselves, as before, to explicit formulae for the flows $D_x$ and $D_t$. These representations are defined as the compatibility conditions for auxiliary linear equations with matrix coefficients
\begin{equation}\label{Psi-eq1}
 \Psi_{n+1}=L_n\Psi_n,\quad
 \Psi_{n,t_i}=U^{(i)}_n\Psi_n,\quad
 \Psi_{n,\xi}=Y_n\Psi_n.
\end{equation}
The compatibility condition for the shift of $n$ and derivations with respect to $t_i$ are
\begin{equation}\label{LU}
 L_{n,t_i}=U^{(i)}_{n+1}L_n-L_nU^{(i)}_n,\quad U^{(i)}_{n,t_j}-U^{(j)}_{n,t_i}=[U^{(j)}_n,U^{(i)}_n].
\end{equation}
In particular,  the lattice (\ref{ux}) and its symmetry (\ref{ut}) are equivalent to the first equation (\ref{LU}) for $i=1,2$, with the matrices of the form
\begin{gather}
\label{LU1}
 L_n=\begin{pmatrix}
  1 & -u_n/\lambda\\ 
  1 & 0
 \end{pmatrix},\quad 
 U^{(1)}_n=\begin{pmatrix}
  u_n & -u_n\\ 
  \lambda & u_{n-1}-\lambda
 \end{pmatrix},\\
\label{U2}
 U^{(2)}_n=\begin{pmatrix}
  u_n(\lambda+u_{n+1}+u_n+u_{n-1}) & -u_n(\lambda+u_{n+1}+u_n)\\ 
  \lambda(\lambda+u_n+u_{n-1}) & -\lambda^2-\lambda u_n+u_{n-1}(u_n+u_{n-1}+u_{n-2})
 \end{pmatrix},
\end{gather}
while the second equation (\ref{LU}) is fulfilled identically. One can check that the negative symmetry corresponds to the matrix
\begin{equation}\label{Yneg}
 Y_n= 
 \frac{1}{\lambda-\alpha}\begin{pmatrix}
  -\alpha y_n-(-1)^n\frac{\beta}{2} & 
  \dfrac{\alpha y^2_n+(-1)^n\beta y_n+\gamma}{y_n+y_{n-1}} \\
  -\lambda(y_n+y_{n-1}) & 
   \lambda(y_n+y_{n-1})-\alpha y_{n-1}+(-1)^n\frac{\beta}{2}
  \end{pmatrix},
\end{equation}
which participate in the remaining compatibility conditions for (\ref{Psi-eq1}):
\begin{equation}\label{LUY}
 L_{n,\xi}=Y_{n+1}L_n-L_nY_n,\quad 
 Y_{n,t_i}=U^{(i)}_{n,\xi}+[U^{(i)}_n,Y_n].
\end{equation}

\begin{proposition}\label{pr:zcr}
For the negative symmetry, equations (\ref{neg}) and (\ref{uy}) are equivalent to the first equations (\ref{LUY}); equations (\ref{yx}) and (\ref{yt}), which define the extensions of derivations $D_x$ and $D_t$ on the variables $y_n$, are equivalent to the second equation (\ref{LUY}) for $i=1,2$.
\end{proposition}

The calculation of the $U^{(i)}_n$ matrices corresponding to higher symmetries (\ref{Di}) amounts to recurrent relations which are equivalent to the recursion operator (\ref{R}). This calculation turns out to be quite simple by using the expansion established in the Proposition \ref{pr:negR}.

\begin{proposition}\label{pr:U}
The matrices $U^{(i)}_n$ from equations (\ref{Psi-eq1}) are of the form
\begin{equation}\label{Ui}
 U^{(i)}_n=\lambda^{i-1}H^{(0)}_n+\lambda^{i-2}H^{(1)}_n+\dots+H^{(i-1)}_n,\quad i=1,2,\dots
\end{equation}
where $H^{(i)}_n$ are the matrices (also depending on $\lambda$)
\begin{equation}\label{Hi}
 H^{(i)}_n=
 \begin{pmatrix}
   h^{(i+1)}_n & -u_n\bigl(h^{(i)}_{n+1}+h^{(i)}_n\bigr) \\
   \lambda\bigl(h^{(i)}_n+h^{(i)}_{n-1}\bigr) & -\lambda\bigl(h^{(i)}_n+h^{(i)}_{n-1}\bigr)+h^{(i+1)}_{n-1}
 \end{pmatrix}
\end{equation}
with the polynomials $h^{(i)}_n$ defined by relations (\ref{hrec}).
\end{proposition}
\begin{proof}
We use the interpretation of the negative symmetry as the generating function for the flows $D_i$. Let $y_n$ be the formal solution (\ref{hseries}) of equation (\ref{uyhom}) with $\beta=0$ and $\gamma=-\alpha/4$. Changing the entry $(1,2)$ of the matrix (\ref{Yneg}) in virtue of this equation yields
\[
 Y_n= -\frac{1}{\lambda-\alpha}\begin{pmatrix}
  \alpha y_n & -u_n(y_{n+1}+y_n) \\
  \lambda(y_n+y_{n-1}) & -\lambda(y_n+y_{n-1})+\alpha y_{n-1}
 \end{pmatrix}.
\] 
Expanding this in a series in $\alpha$, we obtain, using the notation (\ref{Hi}),
\[
 Y_n=-\frac{1}{\lambda-\alpha}\left(
  \frac{\alpha}{2}I+H^{(0)}_n+\frac{1}{\alpha}H^{(1)}_n+\frac{1}{\alpha^2}H^{(2)}_n+\dots\right)
\]
which is equivalent to
\[
 Y_n+\frac{\alpha}{2(\lambda-\alpha)}I = \frac{1}{\alpha}U^{(1)}_n+\frac{1}{\alpha^2}U^{(2)}_n+\dotso.  
\]
The term with the unit matrix $I$ plays no role, since it is canceled in equations (\ref{LUY}), and the comparison with (\ref{tseries}) completes the proof.
\end{proof}

\section{Construction of the non-autonomous reductions}\label{s:nar}

Recall, that if the equation $u_{n,\eta}=g_n$ is an arbitrary symmetry of the lattice equation (\ref{ux}) then the stationary equation $g_n=0$ defines a constraint consistent with this equation. Indeed, the definition of a symmetry implies
\[
 D_x(g_n)=(u_n(u_{n+1}-u_{n-1}))_\eta= u_n(g_{n+1}-g_{n-1})+(u_{n+1}-u_{n-1})g_n,
\] 
and this vanishes in virtue of equations $g_n=0$. In particular, the stationary equation for a linear combination of flows $D_i$ (\ref{Di}) leads to constraints of the form
\[
 P(R)(u_{n,x})=0
\]
where $P(R)$ is a polynomial with constant coefficients on the recursion operator $R$. If the degree of $P$ is equal to $r-1$, then this constraint involves the first $r$ flows $D_i$ and leads to a $(2r-1)$-point difference equation of the form
\[
 \mu_rh^{(r)}_n+\mu_{r-1}h^{(r-1)}_n+\dots+\mu_1h^{(1)}_n = \delta-(-1)^n\varepsilon,
\]
where the integration constants on the right hand side belong to the kernel of the operator $T-T^{-1}$. Such equations serve as a discrete analogue of the Novikov equations from the KdV theory and define finite-gap solutions of the Volterra lattice.

One can add negative symmetries to the linear combination, but this gives nothing new, since stationary equations of the form
\begin{equation}\label{constr1}
 \bigl(\widetilde P(R) +\nu_1(R-\alpha^1)^{-1} +\dots +\nu_m(R-\alpha^m)^{-1}\bigr)(u_{n,x}) =0,
\end{equation}
where $\nu_j$ and $\alpha^j$ are arbitrary constants, can be brought to the form $P(R)(u_{n,x})=0$ by applying the operator product of $R-\alpha^j$.

The situation becomes different for the Painlev\'e type reductions that include symmetries from the additional subalgebra. These symmetries are generated by the recursion operator from the classical scaling symmetry
\begin{equation}\label{scaling}
 u_{n,\tau}= u_n+\sum it_iu_{n,t_i}=u_n(T-T^{-1})\left(\frac{n}{2}+\sum it_ih^{(i)}_n\right)
\end{equation}
where the polynomials $h^{(i)}_n$ are assumed to be homogeneous of degree $i$. In this formula, the sum is taken over an arbitrary finite subset of the flows $D_i$, for which one wish to keep the commutativity property. With some inaccuracy in the notation, we use independent variable $\tau$ for the flow corresponding to any such subset. For example, if we are only interested in the derivatives (\ref{ux}) with respect to $x=t_1$ and (\ref{ut}) with respect to $t=t_2$, then the scaling symmetry is defined as
\begin{equation}\label{scalingxt}
\begin{aligned}
 u_{n,\tau}&= u_n+xu_{n,x}+2tu_{n,t}\\
  &= u_n(T-T^{-1})\Bigl(2tu_n(u_{n+1}+u_n+u_{n-1})+xu_n+\frac{n}{2}\Bigr).
\end{aligned}
\end{equation}
The most general constraints that can be obtained by use of the full symmetry algebra are stationary equations for linear combinations of the flows (\ref{Di}) and the flows
\begin{equation}\label{tDi}
 u_{n,\tau_i}= R^{i-1}(u_{n,\tau}),\quad i=1,2,\dotsc,
\end{equation}
that is, equations of the form
\begin{equation}\label{constr2}
 P(R)(u_{n,x})+Q(R)(u_{n,\tau})=0
\end{equation}
with polynomials $P$ and $Q$. The problem, however, is that the additional flows (\ref{tDi}) are more complicated than (\ref{Di}). The only local flows are the scaling symmetry (\ref{scaling}) itself and the next flow corresponding to the master-symmetry:
\[
 u_{n,\tau_2}= u_n((n+3)u_{n+1}+u_n-nu_{n-1})+\sum it_iu_{n,t_{i+1}}.
\]
It is easy to see that the right hand side of the latter equation does not lie in the image of $u_n(T-1)$, therefore the flow $u_{n,\tau_3}$ is not local. More generally, it can be shown that further application of $R$ leads to the need to introduce a new nonlocality at each step, which makes it difficult to use these symmetries. The main idea of this paper is that it might be more convenient to bring equation (\ref{constr2}) to the form
\[
 P(R)Q^{-1}(R)(u_{n,x})+u_{n,\tau}=0.
\]
Here the first term can be expanded into simple fractions, and in the case of general position (polynomial $Q$ without multiple roots), we arrive to a constraint of the form
\begin{equation}\label{constr3}
 \bigl(\widetilde P(R) +\nu_1(R-\alpha^1)^{-1} +\dots +\nu_m(R-\alpha^m)^{-1}\bigr)(u_{n,x})+u_{n,\tau}=0,
\end{equation}
which differs from (\ref{constr1}) by adding the scaling term, that is, the simplest non-autonomous symmetry. Of course, equation (\ref{constr3}) is nonlocal to the same extent as (\ref{constr2}), but in this version all nonlocalities are moved to terms corresponding to the negative symmetries and therefore they can be defined in already familiar and uniform manner. In (\ref{constr3}), all terms lie in the image $u_n(T-T^{-1})$ and after integration we obtain an equation of the form
\[
 \sum^{r_1}_{i=1}\mu_ih^{(i)}_n + \sum^{r_2}_{i=1}it_ih^{(i)}_n 
  +\frac{1}{2}(n-\delta+(-1)^n\varepsilon) +\nu_1y^1_n+\dots+\nu_my^m_n =0.
\]
Moreover, some further simplifications are possible. First, the coefficients $\mu_i$ and $t_i$ differ only in that we consider the former to be constant and the latter to be variable. However, this is only a matter of interpretation. Therefore, both sums can be combined by choosing $r_2\ge r_1$ and replacing $t_i\mapsto t_i-\mu_i/i$. In the case when we are not interested in the dependence on $t_i$, then it should simply be considered as fixed parameter and differentiation with respect to it should not be considered. Secondly, all coefficients $\nu_j$ can be changed arbitrarily, since multiplication of $y^j$ by a constant only leads to a change of parameters $\beta^j$ and $\gamma^j$ in the equation for this variable.

As a result, we arrive at the following wide family of the constraints:
\begin{equation}\label{nar}
\left\{\begin{aligned}
 & u_n(y^j_{n+1}+y^j_n)(y^j_n+y^j_{n-1})=\alpha^j(y^j_n)^2+(-1)^n\beta^jy^j_n+\gamma^j,\qquad j=1,\dots,m,\\
 & 2\sum^r_{i=1}it_ih^{(i)}_n + n-\delta+(-1)^n\varepsilon=y^1_n+\dots+y^m_n
\end{aligned}\right.
\end{equation}
where $h^{(i)}_n$ are homogeneous polynomials on $u_{n-i+1},\dots,u_{n+i-1}$ defined by recurrent relations (\ref{hrec}). As it was noticed in the previous section, we can assume that $\alpha^i\ne\alpha^j$ for $i\ne j$ and $\beta^j=0$ if $\alpha^j\ne0$ (the cancellation of $\beta^j$ by the changes $y^j_n\mapsto y^j_n-(-1)^n\beta^j/(2\alpha^j)$ results only in the change of the parameter $\varepsilon$). 

Thus, equations (\ref{nar}) define the stationary equation $u_{n,\eta}=0$ for the sum of the scaling and negative symmetries
\begin{equation}\label{ueta}
 u_{n,\eta}=2u_{n,\tau}-u_{n,\xi^1}-\dots-u_{n,\xi^m}.
\end{equation}
By construction, the following property is fulfilled.

\begin{proposition}\label{pr:nar}
The system (\ref{nar}) is consistent with the derivations with respect to the parameters $t_1,\dots,t_r$, defined by equations
\begin{equation}\label{narDi}
 u_{n,t_i}=u_n(h^{(i)}_{n+1}-h^{(i)}_{n-1}),\quad
 y^j_{n,t_i}=h^{(i)}_{n,\xi^j},\quad
 u_{n,\xi^j}=u_n(y^j_{n+1}-y^j_{n-1}).
\end{equation}
These derivations are mutually commutative.
\end{proposition}

It is easy to see that the system (\ref{nar}) has a total order of $2(r+m-1)$ with respect to the shifts of $n$, that is, it is equivalent to some non-autonomous (depending on $n$) mapping on $\mathbb{C}^{2(r+m-1)}$. Accordingly, the non-autonomous (depending on $t_i$) systems of ODEs that arise when the flows $D_1,\dots,D_r$ are restricted due to the constraint equations have the same order. Notice that the order drops by 2 on the hyperplane $t_r=0$ which is special for these flows (in the simplest case $r=m=1$ this leads to an explicit solution \cite{Adler_Shabat_2019}, see also \cite{Adler_2020b}, where such a reduction of dimension was considered for the KdV case).

Notice that the matrices $L_n$ and $U^{(i)}$ from the representations (\ref{LU}) are homogeneous assuming that the spectral parameter $\lambda$ has the same weight as $u_n$. This means that under the scaling, $\lambda$ is transformed in the same way as $u_n$, and at the level of linear problems we have the equation
\[
 \Psi_{n,\tau}+\lambda\Psi_{n,\lambda}=V_n\Psi_n,\quad V_n=\sum it_iU^{(i)}_n.
\]
From here, it follows that the scaling symmetry admits the zero curvature representation with the derivative with respect to the spectral parameter:
\[
 L_{n,\tau}+\lambda L_{n,\lambda}= V_{n+1}L_n-L_nV_n,\quad
 U^{(i)}_{n,\tau}+\lambda U^{(i)}_{n,\lambda}= V_{n,t_i}+[V_n,U^{(i)}_n].
\]
The isomonodromic Lax representation for our constraint is obtained by summation of the matrices corresponding to the flows involved in (\ref{ueta}).

\begin{proposition}\label{pr:nar-Lax}
The following matrix relations are fulfilled by virtue of the system (\ref{nar}) and derivations $D_i$ with respect to the parameters $t_1,\dots,t_r$:
\begin{equation}\label{LW}
 2\lambda L_{n,\lambda}=W_{n+1}L_n-L_nW_n,\quad
 W_{n,t_i}=2\lambda U_{n,\lambda}+[U^{(i)}_n,W_n]
\end{equation}
where
\begin{equation}\label{W}
 W_n=2\sum^r_{i=1}it_iU^{(i)}_n+\sum^m_{j=1}\frac{W^j_n}{\lambda-\alpha^j},
\end{equation}
\[
 W^j_n= 
 \begin{pmatrix}
  \alpha^jy^j_n+(-1)^n\frac{\beta^j}{2} & 
  -\dfrac{\alpha^j(y^j_n)^2+(-1)^n\beta^jy^j_n+\gamma^j}{y^j_n+y^j_{n-1}} \\
  \lambda(y^j_n+y^j_{n-1}) & 
  -\lambda(y^j_n+y^j_{n-1})+\alpha^jy^j_{n-1}-(-1)^n\frac{\beta^j}{2}
  \end{pmatrix},
\]
and the matrices $U^{(i)}_n$ are defined in the Proposition \ref{pr:U}.
\end{proposition}

It should be clarified that the converse is not true: equations (\ref{LW}) do not give the system (\ref{nar})  itself, but only a consequence from it, which is explained by the fact that the matrix (\ref{W}) contains (according to formulae (\ref{Ui}), (\ref{Hi}) and (\ref{hrec})) variables $u_{n-r},\dots,u_{n+r-1}$ which are not independent as one can see from the last equation (\ref{nar}). However, this is easy to fix: in order to obtain a matrix representation which is exactly equivalent to (\ref{nar}), we only need to transform $W_n$ by eliminating the variables $u_{n-r}$ and $u_{n+r-1}$ by use of this equation. We restrict ourselves to an example for the system (\ref{nar-xt}) from the Introduction, which corresponds to $r=2$. For it, the transformed matrix $W_n$ is of the form
\[
 W_n=\begin{pmatrix}
  4t\lambda u_n-n-1-(-1)^n\varepsilon & 
  4tu_n(u_{n-1}-\lambda)+n-\delta+(-1)^n\varepsilon) \\
  4t\lambda(\lambda+u_n+u_{n-1})+2\lambda x & 
  -4t\lambda(u_n+\lambda)-2x\lambda -n+(-1)^n\varepsilon
 \end{pmatrix}
 +\sum^m_{j=1}\frac{W^j_n}{\lambda-\alpha^j}
\]
where
\[
 W^j_n=
 \begin{pmatrix}
  \lambda y^j_n+(-1)^n\frac{\beta^j}{2} & 
  -(\lambda-\alpha^j)y^j_n -\dfrac{\alpha^j(y^j_n)^2+(-1)^n\beta^jy^j_n+\gamma^j}{y^j_n+y^j_{n-1}} \\
  \lambda(y^j_n+y^j_{n-1}) & 
  -\lambda y^j_n-(-1)^n\frac{\beta^j}{2}
 \end{pmatrix}.
\]
A direct calculation proves that equations (\ref{LW}) with this matrix are equivalent to the system (\ref{nar-xt}) and the equations (\ref{ux}), (\ref{yx}) and (\ref{ut}), (\ref{yt}) for $x$- and $t$-derivatives.

\section{Darboux--B\"acklund transformations}\label{s:DBT} 

In this section we consider the generic case $\alpha^j\ne0$ which makes possible to set $\beta^j=0$. In addition, we apply the change $\gamma^j=-\alpha^j(\omega^j)^2$, so that our constraint takes the form
\begin{equation}\label{nar'}
\left\{\begin{aligned}
 & u_n(y^j_{n+1}+y^j_n)(y^j_n+y^j_{n-1})=\alpha^j\left((y^j_n)^2-(\omega^j)^2\right),\\
 & 2\sum^r_{i=1}it_ih^{(i)}_n + n-\delta+(-1)^n\varepsilon=y^1_n+\dots+y^m_n.
\end{aligned}\right.
\end{equation}
Let us show that this system admits the B\"acklund transformations, which allow one to change any of the parameters $\omega^j$ by an integer. The transformation $B_k$ changes only one parameter $\omega^k$ and $\varepsilon$ (other parameters $\omega^j$, as well as $\alpha^j$, $\delta$ and $t_i$ do not change):
\begin{equation}\label{y.Bk}
 B_k:\quad
  \left\{\begin{aligned}
  & \tilde\omega^k=1-\omega^k,\quad \tilde\varepsilon=-\varepsilon,\\
  & \tilde u_n= u_n\frac{(y^k_{n+1}-\omega^k)(y^k_{n-1}+\omega^k)}{(y^k_n)^2-(\omega^k)^2},\\
  & \tilde y^j_n= \frac{1}{\alpha^j-\alpha^k}
     \biggl((\alpha^j+\alpha^k)y^j_n 
           -\alpha^k(y^k_n+\omega^k)\frac{y^j_{n+1}+y^j_n}{y^k_{n+1}+y^k_n}\\
  &\qquad\qquad\qquad\qquad -\alpha^k(y^k_n-\omega^k)\frac{y^j_n+y^j_{n-1}}{y^k_n+y^k_{n-1}}\biggr),\quad j\ne k,\\
  & \tilde y^k_n= -\sum_{j\ne k}\tilde y^j_n
     +2\sum^r_{i=1}it_ih^{(i)}_n(\tilde u_{n-i+1},\dots,\tilde u_{n+i-1})+n-\delta-(-1)^n\varepsilon.
  \end{aligned}\right.
\end{equation} 
In addition, there are transformations which act trivially on the varables $u_n$ and $y^j_n$, and only change the sign of a single parameter:
\[
 A_k:\quad \tilde\omega^k=-\omega^k.
\]
It is easy to see that $A_kB_k:\omega^k\mapsto\omega^k-1$ and $B_kA_k:\omega^k\mapsto\omega^k+1$. 

\begin{proposition}\label{prop:Bk}
The transformations $A_k$ and $B_k$ preserve the form of the system (\ref{nar'}), are consistent with the derivations (\ref{narDi}) and satisfy the commutation relations
\begin{equation}\label{ABgroup}
 A^2_k=\operatorname{id},~~ A_jA_k=A_kA_j,~~ 
 B^2_k=\operatorname{id},~~ B_jB_k=B_kB_j,~~ A_jB_k=B_kA_j,~~ j\ne k.
\end{equation}
The group generated by $A_1,B_1,\dots,A_n,B_n$ is isomorphic to ${\mathbb Z}^n_2\times{\mathbb Z}^n$.
\end{proposition}

The formal proof is rather tedious, and we outline only the main steps. To derive the $B_k$ transformations, we start from the general Darboux transformation for the Volterra lattice. It is defined as
\[
 \widetilde\Psi_n=\rho(\lambda)M_n\Psi_,\quad
 M_n=\begin{pmatrix}
  \lambda & -\lambda\mu f_n\\    
  -\dfrac{1}{f_n-1} & \dfrac{\lambda f_n}{f_n-1} 
 \end{pmatrix}
\]
where $\rho(\lambda)$ is an arbitrary scalar factor. The compatibility condition with the basic equation $\Psi_{n+1}=L_n\Psi_n$ reads $\tilde L_nM_n=M_{n+1}L_n$ and it is equivalent to relations
\begin{equation}\label{uf}
 u_n= -\mu(f_{n+1}-1)f_n,\quad \tilde u_n= -\mu f_{n+1}(f_n-1).
\end{equation}
The compatibility condition with equation $2\lambda\Psi_{n,\lambda}=W_n\Psi_n$ is 
\begin{equation}\label{MW}
 2\lambda M_{n,\lambda}+2\lambda\frac{\rho_\lambda}{\rho}M_n=\widetilde W_nM_n-M_nW_n.
\end{equation}
Here $W_n$ is the matrix of the form (\ref{W}). By using the relations (\ref{Yneg}) and (\ref{Ui}), (\ref{Hi}) for the matrices involved in $W_n$, as well as the last equation of the system (\ref{nar}), it is easy to prove that $\operatorname{tr}W_n$ is independent of the field variables and therefore it does not change under the transformation. Hence it follows that it should be $(\det\rho M_n)_\lambda=0$, which means that the normalization factor at $M_n$ should be chosen equal to
\[
 \rho(\lambda)=\lambda^{-1/2}(\lambda-\mu)^{-1/2}.
\]
Under this choice, the left hand side of (\ref{MW}) contains a pole at the point $\lambda=\mu$, and since the poles in the right hand side are exhausted by the points $\alpha^1,\dots,\alpha^m$, it follows that the parameter $\mu$ must coincide with one of them. Let $\mu=\alpha^k$, then it follows, by comparing relations (\ref{uf}) with the first equation (\ref{nar'}), that $f_n$ can be taken equal to
\[
 f_n=\frac{y^k_n-\omega^k}{y^k_n+y^k_{n-1}},
\]
and this also gives the expression for $\tilde u_n$ (in fact, these are the same discrete Miura transformations as in (\ref{Miura}), up to a linear change of $f_n$). Further, calculating the residues at the points $\alpha^j$ brings to the relations
\[
 \bigl(\widetilde W^j_nM_n-M_nW^j_n\bigr)\Big|_{\lambda=\alpha^j,\mu=\alpha^k}=0,\quad j\ne k,\qquad
 \bigl(\alpha^k M_n+\widetilde W^k_nM_n-M_nW^k_n\bigr)\Big|_{\lambda=\mu=\alpha^k}.
\]
For $j\ne k$, solving these equations gives formulae for $\tilde y^j_n$ and also shows that $\omega^j$ does not change. The relation for $j=k$ turns out to be degenerate: from it we only find $\tilde\omega^k=1-\omega^k$, but $\tilde y^k_n$ cannot be expressed. However, since the transformation formulae for all other variables are already established, $\tilde y^k_n$ is found from the assumption that the last equation (\ref{nar'}) is not changed. After this, it remains only to check (and this is the most difficult part of calculations) that the whole equation (\ref{MW}) turns into identity when $\tilde\varepsilon=-\varepsilon$ is chosen. Consistency with the continuous evolutions and the group identities are also verified by direct calculations at the level of the corresponding matrix representations.

\section{Conclusion}

In this paper we analyzed, for the Volterra hierarchy, the reductions of the string equations type (\ref{constr2}), that is, the stationary equations involving symmetries from the additional subalgebra. It was demonstrated that in the generic case (polynomial $Q$ without multiple roots), such reductions are uniformly written by using nonlocalities corresponding to negative symmetries of the lattice equation. Negative symmetry is itself a very interesting object, at least from the point of view of the study of symmetry algebra, but it is possible that it is also important as an independent integrable equation (recall that the negative symmetry of the KdV equation is point equivalent to the famous Camassa--Holm equation).
 
The case of multiple roots, omitted in the paper, is probably less interesting, but it requires a separate study. Within the framework of the proposed approach, it correspond to higher negative symmetries of the form $(R-\alpha)^{-i}(u_{n,x})$, however it is not clear whether this gives any advantage comparing with the original form of the equation (\ref{constr2}).

We have not touched upon the questions of the analytic properties of the solutions of the constructed reductions. Small-dimensional examples studied in \cite{Fokas_Its_Kitaev_1991, Adler_Shabat_2019}, as well as similar results on string equations for the KdV equation (see, in particular, \cite{Suleimanov_1994, Adler_2020b}) allow us to expect that the reductions under study admit special solutions important for physical applications, but their isolation and study remains a very difficult open problem.

\subsection*{Acknowledgements}

The work was done at Ufa Institute of Mathematics with the support by the grant \#21-11-00006 of the Russian Science Foundation, https://rscf.ru/project/21-11-00006/.

\label{lastpage}
\end{document}